\documentclass[twoside,a4paper,12pt]{amsart}
\usepackage[active]{srcltx}%para activar la relacion inversa con el dvi en Kile (solo linux)
\usepackage{epsfig}

\usepackage{times}
\usepackage{dsfont}
\usepackage[bookmarksnumbered]{hyperref}

\renewcommand{\contentsline}[3]{\csname new#1\endcsname{#2}{#3}}
\newcommand{\newchapter}[2]{\bigskip\hbox to \hsize{\vbox{\advance\hsize by -.5cm\baselineskip=12pt\parfillskip=0pt\leftskip=2cm\noindent\hskip -2cm #1\leaders\hbox{.}\hfil\hfil\par}$\,$#2\hfil}}
\newcommand{\newsection}[2]{\medskip\hbox to \hsize{\vbox{\advance\hsize by -.5cm\baselineskip=12pt\parfillskip=0pt\leftskip=2.5cm\noindent\hskip -2cm #1\leaders\hbox{.}\hfil\hfil\par}$\,$#2\hfil}}
\newcommand{\newsubsection}[2]{\medskip\hbox to \hsize{\vbox{\advance\hsize by -.5cm\baselineskip=12pt\parfillskip=0pt\leftskip=3.5cm\noindent\hskip -2cm #1\leaders\hbox{.}\hfil\hfil\par}$\,$#2\hfil}}

\title{Collapse of spherical charged anisotropic fluid spacetimes}
\author[F.\ Cipolletta, R.\ Giamb\`o]{Federico Cipolletta \and Roberto Giamb\`o}
\address{\begin{tabular}{l}
Scuola di Scienze e Tecnologie \\
Universit\`a di Camerino \\
Italy \\
\texttt{roberto.giambo@unicam.it} \\
\end{tabular}
}

\begin{document}

% Theorems and such

\theoremstyle{plain}\newtheorem{teo}{Theorem}[section]
\theoremstyle{plain}\newtheorem{prop}[teo]{Proposition}
\theoremstyle{plain}\newtheorem{lem}[teo]{Lemma}
\theoremstyle{plain}\newtheorem{cor}[teo]{Corollary}
\theoremstyle{definition}\newtheorem{defin}[teo]{Definition}
\theoremstyle{remark}\newtheorem{rem}[teo]{Remark}
\theoremstyle{definition}\newtheorem{example}[teo]{Example}
\theoremstyle{remark}\newtheorem{step}{\bf Step}
\theoremstyle{plain}\newtheorem*{teon}{Theorem}

%%%%%

\begin{abstract}
A class of spherical collapsing exact solutions with electromagnetic charge is derived. This class of solutions -- in general anisotropic -- contains however  as a particular case the charged dust model already known in literature. Under some regularity assumptions that in the uncharged case give rise to naked singularities, it is shown that the process of shell focusing singularities avoidance -- already known for the dust collapse \cite{kra} -- also takes place here, determing shell crossing effects or a completely regular solution.
\end{abstract}

\maketitle

\section{Introduction}
The study of charged solutions to Einstein's equations is certainly an interesting topic since the beginning of Relativity, and the possibility for an astrophysical object to possess an appreciable amount of electric charge has been debated for a long time \cite{Edd, Glenn, Cue}. Anisotropic charged solutions were studied in \cite{DiP}, where relevance of charge effects is highlighted, and, under self--similarity assumption, in \cite{Barr}. Many numerical studies on the subject exist -- see e.g. \cite{Ghezzi} -- even outside the realm of relativistic fluids, such as charged  scalar fields \cite{oren}. 

In particular, charged collapsing object have attracted 
researchers' attention, for different reasons. On one side, if charge addition prevents the collapsing solution to become singular, Reissner--Nordstr\"{o}m regular interiors could be modeled consequently \cite{Fayos, Rah}. On the other side, if a singularity forms, it is extremely intriguing to investigate the charge effects to the causal behavior of the solution, in particular if the singularity is completely hidden inside a trapped region or a naked singularity exists,  in the latter case providing counterexamples to Cosmic Censorship Conjecture. 

In \cite{kra} the charged spherical dust cloud collapse 
was studied to find that shell focusing singularities are avoided in general, obtaining either shell crossing singularities or completely regular solutions, depending on  mass--to--charge ratio. In the same paper it was conjectured that shell crossing effects could be averted if spatial gradient of pressure remains nonzero. In the present paper we determine analytically a class of charged spherical exact solutions undergoing gravitational collapse, and study their final state. These solutions are in general anisotropic, but contain the charged dust model studied in \cite{kra} as a limit case. We find that shell focusing singularity formation is avoided for non central shells even for this wider class of solutions, and the solution either remains completely regular, or shell crossing singularity is developed. 

The paper is organized as follows. Section \ref{sec:model} is devoted to introducing the model -- in particular, in subsection \ref{sec:addQ} the procedure of charge addition to a spherically symmetric solution is sketched, while the class of (uncharged) background models is introduced in subsection \ref{sec:ns} -- see also the Appendix for the full list of conditions satisfied by these models. The new (charged) solutions are derived in Section \ref{sec:chargedsol} and their relevant properties are studied  in Section \ref{sec:sing}, whereas conclusions are left for the final section.

\section{Charged spherical elements}\label{sec:model}

In this section we introduce relativistic charged models in general, and recall the class of model that we will perturb in next section using an electromagnetic charge.

\subsection{Addition of an electromagnetic charge}\label{sec:addQ}
We start from a spherical material described by the line element, written in comoving coordinates $(t,r,\theta,\phi)$,
\begin{equation}\label{eq:g}
\mathrm ds^2=-e^{2\nu(t,r)}\,\mathrm dt^2+e^{2\lambda(t,r)}\,\mathrm dr^2+R(t,r)^2\,\mathrm d\Omega^2,
\end{equation}
carrying a non-vanishing electric charge density $\rho_e$. We introduce the skew-symmetrc Maxwell tensor $F^{\mu\nu}$ and, due to spherical symmetry and the absence of magnetic monopole \cite{kra,mag2}, the only nonzero
component is $F^{01}$. From the Maxwell's equations
\begin{equation}\label{eq:maxwell}
F^{\mu\nu}_{\quad;\nu}=4\pi\rho_e u^{\mu},\qquad F_{[\mu\nu;\rho]}=0,
\end{equation} 
being $u$ a timelike unit vector field, we find that
\begin{equation}\label{eq:F01}
F^{0 1} = \frac{Q e^{-(\nu+\lambda)} }{R^2},
\end{equation}
where $Q=Q(r)$ is arbitrary and is related to the charge density by the relation
\begin{equation}
\label{eq:Qprime}
Q'(r)= 4\pi{\rho}_e e^{\lambda} R^2.
\end{equation}
Let us consider Einstein equations $G=8\pi T$
for the line element \eqref{eq:g}, where $T=T_{\text{mat}}+T_{\text{em}}$ is the sum of the contributions of the stress energy tensor induced by the material structure and the stress energy electromagnetic tensor 
\begin{equation}
\label{eq:Tem}
(T_{\alpha\beta})_{\text{em}} = - \frac{1}{4 \pi} (F_{\alpha}^{\mu} F_{\mu \beta} + \frac{1}{4} g_{\alpha \beta} F_{\mu \nu} F^{\mu \nu}),
\end{equation}
that using \eqref{eq:F01} 
is found to satisfy
\begin{equation}
\label{eq:Tem-diag}
(T^{\alpha}_{\beta})_{\text{em}} = \text{diag} \left( - \frac{Q^2(r)}{8 \pi R^4}, - \frac{Q^2(r)}{8 \pi R^4}, \frac{Q^2(r)}{8 \pi R^4}, \frac{Q^2(r)}{8 \pi R^4} \right),
\end{equation}
and therefore is diagonal when written as a $(1,1)$--tensor. The above procedure allows to add a charge to a given spherical model.

\subsection{Spherical elastic--solid models}\label{sec:ns}
In the following we are going to describe the spherical model that will be perturbed with an electromagnetic charge following the above scheme. The model considered here has been found and fully described in \cite{gia,gia2}, but for sake of completeness we briefly recall here the basic properties.

We consider a source of gravitational field given by an elastic material under isothermal conditions. The stress energy tensor in comoving coordinates is diagonal,
\begin{equation}
\label{eq:Tmat}
(T^{\alpha}_{\beta})_{\text{mat}} = \text{diag} \left( -\epsilon(t,r),p_r(t,r),p_t(t,r),p_t(t,r)\right),
\end{equation}
where $\epsilon$ is the energy density and $p_r$ and $p_t$ are the radial and the tangential pressure, respectively. As described in \cite{mag,gia}, the additional information to close the system of Einstein equations are supplied by a constitutive function $w=w(r,R,\eta)$ where $\eta=e^{-\lambda}$. Note that $w$ is a function of the spatial coordinates and of the space--space part of the metric, not depending on the angular coordinates $\theta$ and $\phi$ for symmetry reasons. 
The quantities in \eqref{eq:Tmat} are related to $w$ by the identities 
\begin{equation}
\label{eq:wT}
\epsilon=\rho\,w,\qquad p_r=2\rho\eta\frac{\partial w}{\partial\eta},\qquad p_t=-\frac12\rho R\frac{\partial w}{\partial R},
\end{equation}
where
$$
\rho=\frac{\sqrt\eta}{8\pi E(r) R^2}
$$
is the matter density and $E(r)$ is an arbitrary function. 

This framework suggests, as a more 
convenient way to describe the metric, the use of area--radius coordinates $(r,R,\theta,\phi)$, in such a way that the line element becomes \cite{ori}
\begin{equation}
\label{eq:g2}
\mathrm ds^2 = -A(r,R)\, \mathrm dr^2 -2B(r,R)\, \mathrm dR\, \mathrm dr -C(r,R)\,\mathrm dR^2 + R^2\, \mathrm d \Omega^2.
\end{equation}
In \cite{gia,gia2} it is described how, if we choose a $\eta^{-1/2}$--linear constitutive function,
\begin{equation}
\label{eq:w}
w(r,R,\eta)=E(r)\left(\frac{\Psi_{,r}(r,R)}{Y(r,R)}+\Psi_{,R}(r,R)\eta^{-1/2}\right),
\end{equation}
determined by the choice of two free functions $\Psi(r,R)$ and $Y(r,R)$ -- constrained at most to fulfil some reasonability conditions that we will review in a moment -- then Einstein field equations for the metric \eqref{eq:g2} takes the form

\begin{equation}\label{eq:dsori}
\text ds^2=-\left(1-\frac{2\Psi}R\right)G^2\text dr^2 + 2 G\,
\frac Yu \text dR\text dr -\frac 1{u^2} \text dR^2 +R^2 \mathrm d\Omega^2, 
\end{equation}
where
\begin{align}
u^2&=Y^2+\frac{2\Psi}R -1,\label{eq:u2}\\
\intertext{and the function $G$ is given in terms of a quadrature:}
G(r,R)&=\int_R^r\frac{1}{Y(r,\sigma)}\left(\frac1u\right)_{,r}(r,\sigma)\,\mathrm d\sigma+\frac{1}{Y(r,r)u(r,r)}.\label{eq:G}
\end{align}
It is quite straightforward to verify that the arbitrary function $\Psi$ is the Misner--Sharp mass $\Psi=\tfrac R2(1-g(\nabla R,\nabla R))$, whereas the $Y$ is related to the comoving description of the matter through the relation $Y=R'\eta$. Also observe that $u G\eta=1$ as one can find translating the line element \eqref{eq:g} into \eqref{eq:dsori}.

Using \eqref{eq:wT} and \eqref{eq:w}, the components of $T_{\text{mat}}$ as functions of $(r,R,\eta)$ are given by
\begin{align}
&\epsilon=\frac{\Psi_{,r}}{4\pi R^2 Y uG}+ \frac{\Psi_{,R}}{4\pi
R^2}
,\label{eq:eps}\\
&p_r=-\frac{\Psi_{,R}}{4\pi R^2},\qquad p_t=-\frac{1}{8\pi
RuG}\left(\frac{\Psi_{,r}}{Y}\right)_{,R} -\frac{\Psi_{,RR}}{8\pi R}.\label{eq:press}
\end{align}
The above class contains, as particular cases, Tolman--Bondi--Lemaitre dust models -- corresponding to $\Psi=\Psi(r)$ and $Y=Y(r)$ -- and vanishing radial pressure models  -- taking $\Psi=\Psi(r)$ and letting $Y$ possibly depending on $R$ too. In \cite{gia2} some particular 
cases separately studied in literature  including the above cited are reviewed.

As said before, in order to describe reasonably a collapsing matter, some conditions must be placed on the arbitrary functions $\Psi(r,R)$ and $Y(r,R)$. Using
the notation from \cite[Definition 2]{gia}, we suppose that they are chosen in such a way that \eqref{eq:dsori} is a \textit{collapsing area-radius separable (ARS) spacetime} -- see Appendix \ref{sec:ARS} for a complete list of assumptions.  This implies, among the others, that w.e.c. is satisfied and that $\Psi$ and $Y$ are positive and regular functions whose Taylor development in a neighborhood of $(0,0)$ is given by
\begin{align}
&\Psi(r,R)=\sum_{i+j=3}\psi_{ij}r^i R^j+\sum_{i+j=3+m}\psi_{ij}r^i R^j+o_{3+m}(r,R) ,\label{eq:psi-reg}\\
&Y(r,R)=1+\sum_{i+j=2}y_{ij}r^i R^j+\sum_{i+j=2+\ell}\psi_{ij}r^i R^j+o_{2+\ell}(r,R).\label{eq:Y-reg}
\end{align}
Note also that in view of the Definition given in \cite{gia} it must be $\psi_{30}>0$ and $\Psi_{,r},\Psi_{,R}\ge 0$.

Finally, in order to obtain a global model, a matching with an external space will be performed at $\Sigma=\{r=r_b\}$, requiring that the fundamental forms of the two metrics at $\Sigma$ coincide  (Israel--Darmois junction conditions). From \eqref{eq:press} we observe that the radial pressure $p_r$ in general does not vanish along $\Sigma$, which as well known is a necessary and sufficient condition to match the solution with a Schwarzschild exterior. In this ore general case, a natural choice 
for the exterior metric is given by the generalized Vaidya  spacetime \cite{ww}
\begin{equation}\label{eq:10}
g_{\mathrm{ext}}=-\left(1-\frac{2\mu(V,S)}{S}\right)\,\mathrm dV^2+2\,\mathrm dV\,\mathrm dS+S^2\mathrm d\Omega^2,
\end{equation}
where $\mu(V,S)$ is an arbitrary (non negative) function.
 Reissner--Nordstr\"{o}m metric is contained here as a particular case ($\mu=\mu_0-\tfrac{\mu_1}r$, with $\mu_i$ constants).

 The immersion of $\Sigma$ in the two spacetimes can be parameterized respectively by $(\sigma,\theta,\phi)\hookrightarrow(r_b,\sigma,\theta,\phi)$ and $(\sigma,\theta,\phi)\hookrightarrow(V(\sigma),S(\sigma),\theta,\phi)$, and junction conditions are found to be
\begin{eqnarray}
S(\sigma)=\sigma,\quad \frac{\mathrm dV(\sigma)}{\mathrm d\sigma}=\frac{1}{u(Y-u)}|_{(r_b,R=\sigma)}\\
\mu(V(\sigma),S(\sigma))=\Psi(r_b,\sigma),\qquad \frac{\partial \mu}{\partial V}(V(\sigma),S(\sigma))=0.
\end{eqnarray} 

\section{Charged elastic models}\label{sec:chargedsol}
Let us consider the spherical model above recalled and let us add a charge density as described in Section \ref{sec:addQ}. Coupling between elasticity and charge gives rise to a field theory where Lagrangian density  of the elastic material \cite{mag}
$
\Lambda_{\text{mat}}=-\sqrt{-g}\epsilon=-\sqrt{-g}\rho w$ -- where $\rho$ denotes the matter density, $\rho=\eta/(4\pi E(r) R^2)$ -- sums up with the Maxwell Lagrangian density $\Lambda_{\text{em}}=-\tfrac14\sqrt{-g}F^\mu_\nu F^\nu_\mu$. Then using Belinfante--Rosenfeld theorem the total stress energy tensor $T$ is now given by the sum $T_{\text{mat}}+T_{\text{em}}$, that in view of \eqref{eq:Tem-diag} and \eqref{eq:Tmat} is diagonal when expressed in comoving coordinates. Moreover, the perturbation given by \eqref{eq:Tem-diag} depends on $(r,R)$ only. This suggests the idea that the new model may belong to the class considered in Section \ref{sec:ns}. In particular, if we search for new functions $\Psi_1$, $Y_1$ such that the new stress energy tensor obeys relations \eqref{eq:eps} and \eqref{eq:press} for these new functions, it is easy to check that we are led to the following definition:
\begin{align}
\Psi_1(r,R)&=\Psi(r,R)-\frac{Q(r)^2}{2R},\label{eq:psi1}\\
Y_1(r,R)&=\frac{(\Psi_1)_{,r}}{\Psi_{,r}}Y=\left(1-\frac{Q(r)Q'(r)}{R\Psi_{,r}(r,R)}\right)Y(r,R).\label{eq:y1}
\end{align}
In other words, the information about the 
charge effect can be completely encoded into the mass and the $Y$ function in such a way that the solution obtained represent a \textit{fictitious} material belonging to the same class, though with a \textit{new} choice of the arbitrary functions. Note that this procedure is possibile since charge in \eqref{eq:Tem-diag} is added to the {\sl off-shell} relations \eqref{eq:eps} and \eqref{eq:press} where $\eta$ is not explicitly depending on $(r,R)$. 

Of course, the new degree of freedom introduced with the charge modifies the constitutive relation between the elements of the total stress energy tensor,  and indeed now we have
\begin{equation}\label{eq:w1}
w_1(r,R,\eta)=w(r,R,\eta)+\frac{Q^2(r)}{2R^2}\eta^{-1/2}
\end{equation}
 where $w(r,R,\eta)$ is given by \eqref{eq:w}.
At this stage, however, a remark is in place: although $\Psi$ and $Y$ satisfy all the properties sketched in Section \ref{sec:ns} and listed in the Appendix  -- especially regularity \eqref{eq:psi-reg} and \eqref{eq:Y-reg} --in view of \eqref{eq:psi1} and \eqref{eq:y1} one cannot expect that the same properties are satisfied by $\Psi_1$ and $Y_1$. Put it in another way, choosing the two arbitrary functions as in \eqref{eq:psi1} and \eqref{eq:y1}  above, the charged model obtained is \textit{not} an ARS collapsing spacetime in the sense of \cite[Definition 2]{gia}, and then results from there do not work here. For instance, we cannot exclude shell crossing singularity formation -- that on the contrary is likely to occur as we will see later. Nevertheless it is easy to see that the new model automatically satisfies the weak energy condition if the old does. Moreover, if we require regularity of the energy density in the center, recalling that $R=r$ at initial time, the condition
$Q^2(r)=O(r^4)$ as $r\to 0$ is needed, as it can be seen using \eqref{eq:eps} together with the facts that $R'=u G Y$ and $R'=1$ at initial time. Assuming regularity of the charge function $Q(r)$ as done for the other arbitrary functions, this amounts to say that 
$Q(r)=q_0 r^{2+p}+o(r^{2+p})$, with $p\ge 0$. 

Moreover, we have to assume that local isotropy at the center continues to hold once the charge is added. In comoving coordinates, using \eqref{eq:Tem-diag}, that would mean that
\begin{equation}
\label{eq70}
\lim_{r \to 0^+} \left( p_{r1} (t,r) - p_{t1} (t,r) \right) = \lim_{r \to 0^+} \left( p_r (t,r) - p_t (t,r) - \frac{Q^2(r)}{4 \pi R^4(t,r)}\right) = 0.
\end{equation}

Supposing that the local isotropy holds for the uncharged model, this means that 
$$\lim_{r\to 0^+} \frac{Q^2(r)}{R^4(t,r)} = 0$$
whenever we tend to the regular center. Recalling again the initial data $R=r$, that implies that $Q(r)=o(r^2)$, and therefore the condition stated above on the behavior of the charge function must be strengthened as follows
\begin{equation}
\label{eq71}
Q(r) = q_0 r^{2+p} + o(r^{2+p}) , \qquad  p \geq 1.
\end{equation}

\section{Collapsing solutions and singularity formation}\label{sec:sing}
In view of the lack of regularity at the center -- see \eqref{eq:psi1} and \eqref{eq:y1} -- one cannot expect that these solutions collapse completely to a shell focusing singularity. To begin, an important preliminary result can be immediately stated.
\begin{prop}\label{thm:noncentral}
Shell focusing singularity cannot form from the collapse of non central shells of the charged elastic models described in Section \ref{sec:chargedsol}, unless a shell crossing singularity forms at an earlier comoving time.
\end{prop}

\begin{proof}
We will actually show that, for each $r_0>0$, the function $R\mapsto Y_1(r_0,R)$ must vanish for some $R\in]0,r_0[$. Since the equation $Y_1=0$ controls shell crossing singularity formation, we will deduce that shell focusing non central singularities cannot form before possibile shell crossing singularities. 

From \eqref{eq:y1}, and recalling from \eqref{eq:Y-reg} that $Y(r,R)$ is strictly positive, the zeroes of $Y_1$ corresponds to zeroes of 
$\Psi_1 ,_r (r,R)= \Psi ,_r (r,R)- \frac{Q(r) Q'(r)}{R}$. Fixed $r=r_0$, this quantity negatively diverges as $R\to 0^+$, which means that the continuous function  $R\mapsto \Psi_1 ,_r (r_0,R)$, that starts positive for $R=r_0$, must vanish for some positive $R<r_0$.
\end{proof}

In view of the above result, non central shells either
develop a shell crossing singularity or they stay regular, collapsing indefinitely to a stable non singular state or even bouncing and finally entering a period of expansion. 

Shell crossing formation corresponds to the zeroes of the function $Y$ (that in comoving coordinates reads as $R'\eta$), while no singularity formation is possible if $\tfrac{2\Psi_1}R-1+Y_1^2$ vanishes -- consider \eqref{eq:u2} restated using $\Psi_1$ and $Y_1$ in place of $\Psi$, $Y$. 

Actually, points of the $(r,R)$ plane where $\tfrac{2\Psi_1}R-1+Y_1^2$ is negative describe a region where the dynamics of the spacetime is not allowed; taking initial data such that this quantity is positive, and studying the evolution of this quantity for a fixed shell $r=r_0$, if the quantity vanishes for some $R=R_u(r_0)$ and attains negative values for smaller values of $R$, this means that the shell labeled $r_0$ cannot collapse beyond the limit value $R=R_u(r_0)$, then indefinitely tends to it or arrives there in a finite amount of comoving time, bouncing back and entering an expanding phase.

In view of the aforesaid, there exists a curve $R_{sc}(r)$ such that, for $r_0>0$, $R_{sc}(r_0)$ is the greatest zero of the function $R\mapsto Y_1(r_0,R)$ between $0$ and $r_0$ (recall that $Y_1(r_0,r_0)>0$ by the initial data choice). Recalling \eqref{eq:y1} for $Y_1$, we have to determine the zeroes of
$(\Psi_1)_{,_r}$ and
using  \eqref{eq:psi-reg} and \eqref{eq71}, it is found explicitly that
\begin{equation}
\label{eq88}
R_{sc}(r) = \frac{(2+p) {q_0}^2}{3 {\psi}_{30}} r^{2p+1} + o(r^{2p+1}).
\end{equation}

Moreover, to investigate the physics of the collapse, we need to study the apparent horizon, implicitly defined by solution of equation
\begin{equation}\label{eq:hor}
R=2\Psi_1(r,R).
\end{equation}
Recalling equations \eqref{eq:psi-reg} and \eqref{eq71} we find the existence of \textit{two} different horizon curves satisfying \eqref{eq:hor}, whose behavior depends on the charge function $Q(r)$. In particular, we have the following situation:
\begin{itemize}
\item when $p>1$ (and $\psi_{30}\ge q_0$):
\begin{equation}
\label{eq:hor-pgreat}
{R_h}_1(r) = \frac{{q_0}^2}{2 \psi_{30}} r^{2p+1} + o(r^{2p+1}),
\quad {R_h}_2(r) = 2 \psi_{30} r^3 + o(r^3).
\end{equation}
\item when $p=1$:
\begin{equation}\label{eq:hor-p1}
{R_h}_1(r) = \left( \psi_{3 0} - \sqrt{{\psi_{3 0}}^2 - {q_0}^2} \right) r^3 + o(r^3),\quad
{R_h}_2(r) = \left( \psi_{3 0} + \sqrt{{\psi_{3 0}}^2 - {q_0}^2} \right) r^3 + o(r^3),
\end{equation}
\end{itemize}
As can be inferred from \eqref{eq:hor-p1}, when $p=1$ and $\psi_{30}<q_0$, no horizon forms. In all other cases, we have two apparent horizons which bound the shell crossing curve $R_{sc}$ \eqref{eq88} one from below and the other from above. 
As said before, vanishing of the function 
\begin{equation}\label{eq:u1}
u_1(r,R):=\sqrt{\frac{2\Psi_1(r,R)}{R}-1+Y_1(r,R)^2}
\end{equation}
could prevent shell crosing formation. From the above analysis, there is certainly a case when this happens, that is when no horizon formation takes place. Indeed, in that case it must be $R>2\Psi_1$ throughout the whole evolution, and then on the curve $R_{sc}$ \eqref{eq88}, the quantity under square root in \eqref{eq:u1} is negative. Therefore, $u_1$ already vanished before $R_{sc}$ which means that the solution remains regular for all times. To gain more information on the behavior of $u_1$ in this case ($p=1,\, \psi_{30}<q_0$), one can study the quantity
$R u_1^2$ along curves $(r,k r)$ for $k\in ]0,1]$, finding that
$$
{R {u_1}^2(r, R)}_{|_{R = kr}} = \varphi(k) r^3 + o(r^3),
$$
 where
\begin{equation}\label{eq:phik}
\varphi (k) = \frac{2A(k)}{(3 \psi_{30} + 2 k \psi_{21} + k^2 \psi_{12})},
\end{equation}
with 
\begin{multline}\label{eq:Ak}
A (k) = -3 q_0^2+\left(k^2 \psi _{12}+2 k \psi _{21}+3 \psi _{30}\right) \cdot\\
\cdot\left(k^3 \psi _{03}+k^3 y_{02}+k^2 \psi _{12}+k^2 y_{11}+k \psi _{21}+k y_{20}+\psi _{30}\right).
\end{multline}
The denominator in \eqref{eq:phik} is positive -- is the leading term of $\Psi_{,r}(r,kr)$ which is positive by definition, see subsection \ref{sec:ns} -- therefore we have to study the sign of $A(k)$ \eqref{eq:Ak}.  
In particular we have $A(1) > 0$ (as it is the initial data) and $A(0)< 0$ (since $\psi_{30}<q_0$). Thus $\varphi(k)$ changes sign for some $k_0\in ]0,1[$, and then the zeroes of $R u_1^2$ defines a curve
$$R_u(r)=k_0 r+o(r)$$ 
which bounds $R_{sc}$ from above as expected.
\begin{figure}
\includegraphics{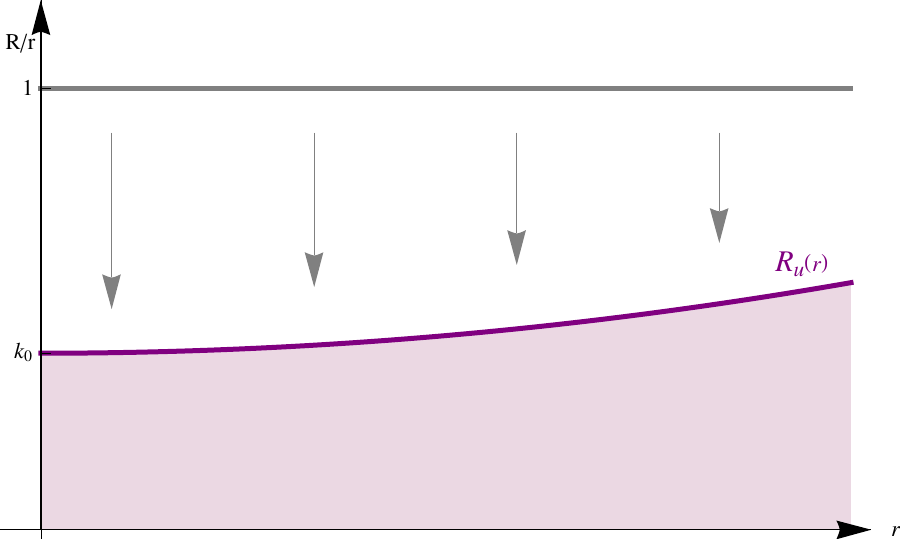}
\caption{When a curve $R_u(r)$ -- where the quantity under square root in \eqref{eq:u1} changes sign -- is defined, there exists a region -- the shaded area below that curve -- that cannot be reached starting the collapse from $R/r=1$.}
\label{fig:1}
\end{figure}
The collapse behavior is sketched in Figure \ref{fig:1}, where a graphic is plotted with $R/r$ on the vertical axis to get rid of the central coordinate singularity. The collapse starts from $R/r=1$ and every shell collapses until it reaches $R=R_u(r)$ where it cannot proceed further since the quantity under square root in \eqref{eq:u1} would be negative. Then, once the shell radius reaches the value $R=R_u(r)$, it starts bouncing back entering an expanding phase. Since $R_u(r)/r\to k_0>0$ as $r\to 0^+$, even the center remains regular for all times.

\begin{rem}\label{rem:Ak}
It can be seen that the quantity $A(k)$ is given by the leading order term of the function $\Psi_{,r}(r,R) R u^2(r,R)-2Q(r)Q'(R)$, evaluated for $R=k r$.
\end{rem}

 Now let us examine the other cases. When  $p=1$ and $\psi_{30}\ge q_0$ we can perform a similar study as before, finding that $R u_1^2(r,kr)=\varphi(k) r^3+o(r^3)$, with $\varphi(k)$ given by \eqref{eq:phik} above, and then its sign is related to the sign of $A(k)$ \eqref{eq:Ak} again. However, in this 
case both $A(0)$ and $A(1)$ are positive, so we have to consider separately the following two subcases:
\begin{itemize}
\item[(a)]
if $A(k)$ admits one root in the interval $]0,1[$, then as before it remains defined a curve $$R_u(r)=k_0 r+o(r)$$ 
which bounds $R_{sc}$ from above, and then the situation is the same as that depicted in Figure \ref{fig:1}.
\item[(b)]
if $A(k)>0$ in $[0,1]$, then we have to evaluate $R u_1^2$ along $R = k r^2$. For sake of simplicity we consider here only the case $\psi_{30}>q_0$ (the special case $\psi_{30}=q_0$ can be discussed with similar methods). It can be found $R u_1^2(r,kr^2)=\varphi(k) r^3+o(r^3)$, where
$$
\varphi (k) = \frac{2 ({\psi_{30}}^2 - {q_0}^2)}{\psi_{30}},
$$
that is positive. Evaluating finally $R u_1^2$ on curves $R = kr^3$, one gets again a behavior of the kind  $R u_1^2(r,kr^2)=\varphi(k) r^3+o(r^3)$, where
\begin{equation}
\label{eq100}
\varphi (k) = \frac{2 ({\psi_{30}}^2 - {q_0}^2)}{\psi_{30}} + \frac{{q_0}^2}{k {\psi_{30}}^2} ({q_0}^2 - {\psi_{30}}^2).
\end{equation}
and therefore the zeroes of $R u_1^2$ describe a curve $R_u(r)$ whose behavior is given by
\begin{equation}
\label{eq101}
R_u (r) = \frac{{q_0}^2}{2 \psi_{30}} r^3 + o(r^3).
\end{equation}
\begin{figure}
\includegraphics{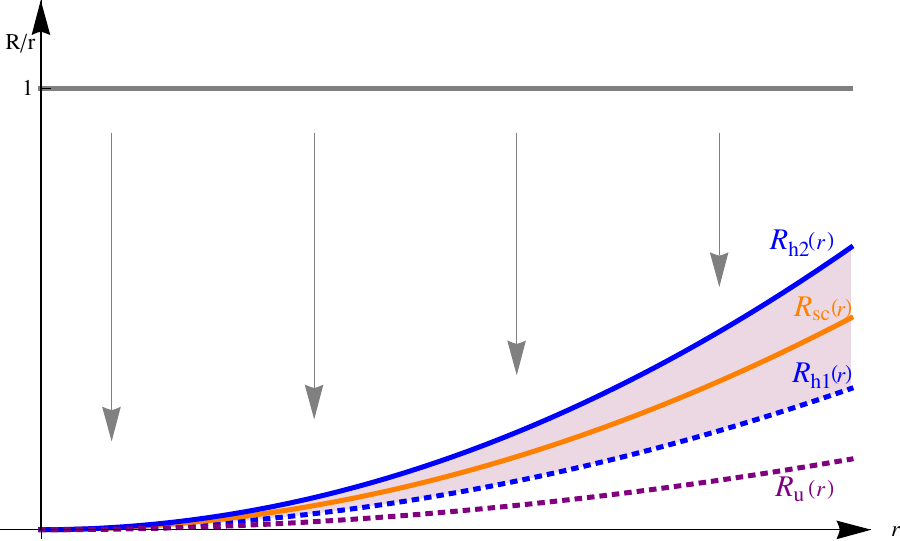}
\caption{When the curve $R_u(r)$ lies below the inner horizon, then a shell crossing singularity -- denoted by the curve $R_{sc}(r)$ -- forms. The shaded area represents the trapped region.}
\label{fig:2}
\end{figure}
Comparing the above equation with \eqref{eq88} we conclude that $R_u(r)<R_{sc}(r)$ for $r>0$ and then in this case a shell crossing singularity develops just after the formation of the outer horizon. The situation is represented in Figure \ref{fig:2}.
\end{itemize}

Finally, case $p>1$ is yet left to be studied. It is easy to check that zeroes of $R u_1^2$ corresponds to solutions to the equation
$$
2\Psi {\Psi ,_r}^2 R_u + {R_u}^2 {\Psi ,_r}^2 \left[ -1 + Y^2 \right] - 2 Q Q' R_u Y^2 \Psi ,_r = Q^2 \left[ {\Psi ,_r}^2 - {Q'}^2 Y^2 \right].
$$
From this, using expansions \eqref{eq:psi-reg}, \eqref{eq:Y-reg} and \eqref{eq71} and comparing the leading term of the left side member with the leading term of the right side member in the above equation leads to  the following behavior for function $R_u$ such that $u_1(r,R_u(r))=0$:
\begin{equation}
\label{eq91}
R_u (r) = \frac{{q_0}^2}{2 {\psi}_{30}} r^{2p+1} + o(r^{2p+1}).
\end{equation}
An immediate comparison with equation \eqref{eq88} shows that $R_u(r)<R_{sc}(r)$ for positive $r$ and then all shells as before undergo a shell crossing formation, and the situation is again represented in Figure \ref{fig:2}.

Summarizing, we have that the behavior of the collapse is driven by the parameter $p$ defined in \eqref{eq71} in the following way:

\begin{teo}\label{thm:fs}
Under the assumptions made in Section \ref{sec:chargedsol} a
shell crossing singularity forms up to the centre when:
\begin{enumerate} 
\item the parameter $p$ defined in \eqref{eq71} is greater than 1, or 
\item $p=1$, $\psi_{30}> q_0$ and the function $A(k)$ defined in \eqref{eq:Ak} is positive $\forall k\in[0,1]$.
\end{enumerate}
In all other cases, provided $\psi_{30}\ne q_0$\footnote{ the highly non generical case $p=1, \psi_{30}=q_0$ can be treated with similar methods as before, just taking into account higher order terms of the arbitrary functions.} the solution remains regular, and no singularity forms up to the center.
\end{teo}

\section{Conclusions}\label{sec:final}
The class of spherical models carrying an electromagnetic charge considered in this paper has been explicitly derived exploiting the description in terms of a state function $w=w(r,R,\eta)$ -- see subsection \ref{sec:ns}. Since the addition of a charge modifies the stress--energy tensor with a quantity depending on $(r,R)$, it is possible to derive the state function for the new solution, seen as a fictitious material. Therefore, one can say that the process of solution charging is -- under a mathematical point of view -- internal to the class of models considered, i.e. those determined by the state function \eqref{eq:w}. Although the charged solution does not satisfy all the regularity properties of the original one, however it is possible to explicitly derive the metric and then to study its properties.

As we have seen, the electromagnetic charge deeply changes the final state of the solution. It is a well established fact indeed that the uncharged solutions here considered always develop a shell focusing singularity that in some case is naked -- depending on Taylor expansion of the arbitrary functions $\Psi(r,R)$ and $Y(r,R)$, see \cite{gia,gia2} for full details. Instead, here the formation of a shell focusing singularity is impossible for non central shells, and otherwise either the solution bounces back and enters an expanding phase, or a shell crossing singularity takes place. 
As is well known shell crossing singularities are Tipler--gravitationally weak \cite{newm} and the spacetime may admit an extension beyond this kind of singularity, possibly developing a shell focusing singularity at a later comoving time. This issue
still remains an open problem in general, that has been partially solved only for particular cases, as Tolman--Bondi models \cite{nol}, where metric can be extended in a natural but \textit{non} unique and, yet bounded, \textit{non} regular way beyond the shell--crossing singularity.

The occurence of the two situations -- either regularity or shell crossing effect -- is determined (Theorem \ref{thm:fs}) by the Taylor development of $\Psi$, $Y$ and the charge function $Q(r)$. Notice that, unlike the charged dust cloud case studied in \cite{kra}, the final state is not determined exclusively by the charge--to--mass ratio at initial time $t=0$, since now the mass depends on $R$ also and at initial time (when $R=r$) we have $\Psi(r,r)=\left(\sum_{i+j=3}\psi_{ij}\right)r^3+o(r^3)$, that in general does not allows to determine the coefficient $\psi_{30}$ needed in Theorem \ref{thm:fs}.

However the avoidance of shell focusing singularities  -- up to spacetime extension as said before --  already found in \cite{kra} remains true for this wider class of (generally anisotropic) solutions, giving a first answer  (in the negative) to the conjecture raised in \cite{kra} that the occurence of shell crossing singularities was determined by the vanishing of the pressure spatial gradient. It emerges that electromagnetic charge completely changes the final state of the collapsing star in this anisitropic case.

\appendix

\section{Collapsing area--radius separable (ARS) spacetimes}\label{sec:ARS}

The models introduced in \cite{gia,gia2} and recalled in Section \ref{sec:ns} completely depend on the choice of two functions $\Psi$ and $Y$. 
An \textit{ARS spacetime} is described by the metric \eqref{eq:dsori}, where \eqref{eq:u2}--\eqref{eq:G} holds, and $\Psi(r,R)$, $Y(r,R)$ are two regular functions. Since we consider initial data such that $R=r$, the spacetime is defined in the set $\mathcal S=\{(r,R)\,:\,0\le r\le r_b,\,0\le R\le r\}$, and $r_b$ defines the junction hypersurface with the external solution -- see end of Section \ref{sec:model} above.  Now we give a complete list of the properties that these functions must satisfy to obtain what in \cite{gia} is termed as a \textit{collapsing} ARS spacetime. 

\begin{itemize}
\item  the arbitrary functions must be positive up to  possible singularity formation, $\Psi(r,R)>0$ and $Y(r,R)>0$, when $R>0$;
\item the weak energy condition (w.e.c.) must hold. A sufficient condition is given by
\begin{equation}\label{eq:wec}
{\Psi_{,r}}\ge 0,\quad\Psi_{,R}\ge 0,\quad \Psi_{,r} \geq \frac R2
Y \left(\frac{\Psi_{,r}}{Y}\right)_{,R} ,\quad \Psi_{,R} \geq
\frac R2 \Psi_{,RR};
\end{equation}
\item regularity and isotropy of the metric at the center of symmetry must be imposed. It can be seen that this is equivalent to the following conditions:
\begin{align}
&Y(0,0)=1,\label{eq:cond2}\\
&\Psi(0,0)=\Psi_{,r}(0,0)=\Psi_{,R}(0,0)=\Psi_{,rr}(0,0)=
\Psi_{,rR}(0,0)=\Psi_{,RR}(0,0)=0;\label{eq:cond1}
\end{align}
\item the initial density is required to be decreasing outwards, which means that
\begin{equation}
\Psi_{,rr}(r,r)+2\Psi_{,rR}(r,r)+\Psi_{,RR}(r,r)-\frac 2r
\left(\Psi_{,r}(r,r)+\Psi_{,R}(r,r)\right)<0\label{eq:cond3};
\end{equation}
\item shell--crossing singularity formation must be excluded. This leads to the (sufficient) conditions
\begin{equation}\label{eq:Hrpos}
\Psi_{,r}(r,R)+Y(r,R)\,Y_{,r}(r,R)\,R\ge 0 ,
\end{equation}
\begin{equation}\label{eq:nosc}
\int_0^r\frac{1}{Y(r,\sigma)}\left(\frac1u\right)_{,r}(r,\sigma)\,\text
d\sigma+\frac{1}{u(r,r)Y(r,r)}> 0\qquad\text{for\ \ }r>0;
\end{equation}
\item shell--focusing singularities must form in a finite amount of comoving time. The mathematical conditions to express this requirement are
\begin{equation}\label{eq:shfoc}
Y_{,r}(0,0)=Y_{,R}(0,0)=0,\quad(\Psi_{,rrr}(0,0),\Psi_{,rrR}(0,0),\Psi_{,rRR}(0,0))\ne(0,0,0).
\end{equation}
\end{itemize}
In particular, regularity of $\Psi$ and $Y$, together with \eqref{eq:cond2}-\eqref{eq:cond1} and \eqref{eq:shfoc} allows to write developments \eqref{eq:psi-reg} and \eqref{eq:Y-reg}. Moreover in \cite{gia,gia2} to ensure the very last condition in \eqref{eq:shfoc}, it is assumed that $\Psi_{,rrr}(0,0)>0$, that in view of \eqref{eq:psi-reg} means  $\psi_{30}>0$.

\end{document}